\documentclass[11pt,a4paper]{amsart}
\pdfoutput=1

\newtheorem{thm}{Theorem}[section]
\newtheorem{prop}[thm]{Proposition}

\newtheorem{lem}[thm]{Lemma}
\newtheorem{hypo}[thm]{Hypothesis}

\theoremstyle{definition}

\newtheorem{ex}[thm]{Example}

\usepackage[square,sort,comma,numbers]{natbib}
\usepackage[english]{babel}
\usepackage{graphicx}
\usepackage[all,cmtip,2cell,graph]{xy}
\usepackage{pstricks}
\usepackage{enumerate}
\usepackage{amsfonts,amssymb,amsmath,pinlabel,array,hhline}
\usepackage{slashed}
\usepackage{tabulary}
\usepackage{fancyhdr}
\usepackage{a4wide}
\usepackage{mathrsfs}
\usepackage{calrsfs}
\usepackage{bbm,dsfont}
\usepackage{mathtools}
\UseAllTwocells

\DeclareSymbolFont{EulerScript}{U}{eus}{m}{n}
\DeclareSymbolFontAlphabet\mathscr{EulerScript}

\newcommand{\R}{\mathbb{R}}
\newcommand{\Z}{\mathbb{Z}}

\newcommand{\C}{\mathbb{C}}
\newcommand{\cov}{\mathrm{Cov}}
\newcommand{\var}{\mathrm{Var}}

\begin{document}

\title{The topological hypothesis for discrete spin models}

\author{David Cimasoni}
\address{University of Geneva, Section of mathematics, CH-1211 Gen\`eve, Switzerland}
\email{david.cimasoni@unige.ch}

\author{Robin Delabays}
\address{University of Applied Science of Western Switzerland, CH-1950 Sion, Switzerland}
\email{robin.delabays@hevs.ch}


\begin{abstract}
The topological hypothesis claims that phase transitions in a classical statistical mechanical system are related to changes in the topology of the level sets of the Hamiltonian. So far, the study of this hypothesis has been restricted to continuous systems. The purpose of this article is to explore discrete models from this point of view. More precisely, we show that some form of the topological hypothesis holds for a wide class of discrete models, and that its strongest version is valid for the Ising model on~$\Z^d$ with the possible exception of dimensions~$d=3,4$.
\end{abstract}

\maketitle

\section{Introduction}\label{sec:intro}

In 1997, a new and unconventional approach to the study of equilibrium phase transitions was suggested by Caiani {\em et al.}~\cite{Cai97}. In a nutshell, the idea of this {\em topological approach\/} is to consider the configuration space~$\Omega_\Lambda$ as a manifold, the Hamiltonian~$H\colon\Omega_\Lambda\to\R$ as a Morse function, and to relate the appearance of a phase transition (understood as a non-analyticity of some thermodynamic function, usually the pressure) to a change in the topology of the manifold~$M_\Lambda(u)=\{\omega\in\Omega_\Lambda\colon H(\omega)\le |\Lambda| u\}$ as the number~$|\Lambda|$ of particles tends to infinity. Originally supported only by numerical evidence~\cite{Cai97,Fra00} and phrased in rather vague terms, this hypothesis was later formulated as a series of conjectures commonly refered to as the {\em topological hypothesis\/}, and some of these conjectures were proven to hold for the mean-field~$XY$-model~\cite{Cas03} and the mean-field~$k$-trigonometric model~\cite{Ang05a} (see also~\cite{Ang05b,Far11,San17}). Furthermore, Franzosi and Pettini proved that for a certain class of models, a topological change within the family of manifolds~$\{M_\Lambda(u)\}_{u\in\R}$ with~$|\Lambda|$ large is a necessary condition for a phase transition to occur~\cite{Fra04,Fra07b} (see also~\cite{KM11,MHK12,GFP18}). However, it soon became clear that the initial hope of this topological approach providing a general description of phase transitions was over-optimistic. Indeed, none of the various incarnations of the topological hypothesis holds true for arbitrary systems~(see e.g.~\cite{Kas06,Ris06,Ang07}). We refer the reader to the beautiful survey~\cite{Kas08a} and references therein for more details (see also~\cite{Buc08}).

To this day, the study of the topological hypothesis has been restricted to continuous models, i.e. models where the manifold~$\Omega_\Lambda$ has positive dimension. However, discrete spaces are (zero-dimensional) manifolds in their own right, so it makes perfect sense to explore the validity of the topological hypothesis for discrete models. This is the aim of the present article.

To be more precise, we study the strongest version of the topological hypothesis, which equates a phase transition at inverse temperature~$\beta_c>0$ with a non-analyticity of the logarithmic density~$\sigma$ of the Euler characteristic of~$M_\Lambda(u)$ at the corresponding energy~$u_c\in\R$ (see Section~\ref{sec:prelim} below). For this statement to make sense, we need this correspondence between inverse temperatures and energies to be one-to-one; in other words, we need equivalence of ensembles to hold (see Section~\ref{sec:ens_equiv}), and the pressure to be differentiable and strictly convex. In our main theorem, we show that under some hypotheses that ensure the occurence of this situation, a slightly modified version of the topological hypothesis holds true (Theorem~\ref{thm:main}). We then apply this result to the ferromagnetic nearest neighbour Ising model on~$\Z^d$, where the full topological hypothesis is shown to hold with the possible exception of dimensions~$d=3,4$.

Obviously, a discrete space is {\em topological\/} in the technical sense of the word, but not so much in the ``intuitive'' sense. For this reason, it is fair to say that there is not much topology left in the topological hypothesis for discrete spaces. These semantic considerations aside, proving the validity of this hypothesis for a wide class of discrete models provides an indisputable argument in favor of the topological approach.

This article is organised as follows. Section~\ref{sec:prelim} contains the definitions and terminology necessary for the statement of the topological hypothesis. In Section~\ref{sec:ens_equiv}, we recall classical results on the equivalence of ensembles for lattice spin models. Finally, in Section~\ref{sec:main}, we relate the function~$\sigma$ to the entropy, we study the strong convexity of the pressure, prove our main result, and illustrate it with the example of the ferromagnetic Ising model.

\subsection*{Acknowledgments} The authors would like to thank Hugo Duminil-Copin, Sacha Friedli and Yvan Velenik for helpful discussions, as well as the anonymous referee for useful comments. DC was partially supported by the Swiss FNS. RD was supported by the SNF AP Energy grant PYAPP2\textunderscore 154275.

\section{The topological hypothesis}
\label{sec:prelim}

In this somewhat dry preliminary section, we recall the definitions and terminology necessary for the statement of the topological hypothesis in the general setting of spin systems, following~\cite{Kas08a,Tou15}. We refer the interested reader to these articles for further details.

\subsection{Thermodynamic equivalence of ensembles}
\label{sub:equiv}

We consider a spin system of a finite set~$\Lambda$ of classical particles. Such a system is characterized by a {\em Hamiltonian\/}
\[
H_\Lambda\colon\Omega_\Lambda\to\R
\]
defined on the {\em configuration space\/}~$\Omega_\Lambda=S^\Lambda$, where~$S$ is some measured space called the {\em spin space\/}. We will denote by~$\rho_\Lambda$ the corresponding product measure on~$\Omega_\Lambda$, with respect to which~$H_\Lambda$ is assumed to be measurable. Given a {\em spin configuration\/}~$\omega\in\Omega_\Lambda$, the quantities~$H_\Lambda(\omega)$ and~$\frac{H_\Lambda(\omega)}{|\Lambda|}$ are called the {\em energy\/} and {\em energy per particle\/} of~$\omega$, respectively.

From this data, two thermodynamic functions can be defined. On the one hand, the {\em pressure\/} is the function of the inverse temperature~$\beta\in\R$ given by
\begin{equation*}
\psi(\beta) \coloneqq \lim_{|\Lambda|\to\infty}-\frac{1}{|\Lambda|}\log\int_{\Omega_\Lambda}e^{-\beta H_{\Lambda}(\omega)}d\rho_\Lambda(\omega)\,.
\end{equation*}
On the other hand, the {\em microcanonical entropy\/} is the function of the energy per particle~$u\in\R$ given by
\begin{equation*}
s(u) \coloneqq \lim_{r\to 0}\lim_{|\Lambda|\to\infty}\frac{1}{|\Lambda|}\log \rho_\Lambda\left\{\omega\in\Omega_\Lambda\colon\frac{H_\Lambda}{|\Lambda|}\in(u-r,u+r)\right\}\,.
\end{equation*}

Under some assumptions (see e.g.~\cite{Rue69}, and Section~\ref{sub:LPS} below), it can be shown that if the function~$s$ exists, then~$\psi$ also exists and is equal to the {\em Legendre-Fenchel transform\/}\footnote{The standard form of the Legendre-Fenchel transform is~$\tilde{s}(\beta)\coloneqq\sup_{u\in\R}\{\beta u-s(u)\}$, which is related to~$s^*$ via~$\tilde{s}(\beta)=-(-s)^*(-\beta)$. Also, the standard definition of the pressure is~$p(\beta)=-\beta^{-1}\psi(\beta)$. Following~\cite{Tou15}, we use these slightly modified conventions to avoid carrying signs around. Note that~$\psi(\beta)$ and~$f(\beta)=\beta^{-1}\psi(\beta)$ are also commonly referred to as the {\em free energy\/}.} of~$s$:
\[
\psi(\beta)=s^*(\beta)\coloneqq\inf_{u\in\R}\{\beta u-s(u)\}\,.
\]
If~$s$ is concave, then the inverse equality~$\psi^*=s$ also holds and {\em thermodynamic equivalence of ensembles\/} is said to occur~\cite{Tou15}.
Assuming further that~$s$ and~$\psi$ are (continuously) differentiable, the functions~$s'$ and~$\psi'$ are inverses of each other. This provides a one-to-one correspondence between inverse temperatures and energies per particle.

\subsection{The topological hypothesis}
\label{sub:TH}

The system is said to undergo a \emph{phase transition} at inverse temperature~$\beta>0$ if the pressure~$\psi$ is not smooth at~$\beta$, i.e. if it is not infinitely many times differentiable at~$\beta$. Following a slightly outdated terminology, we will say that this phase transition is of {\em order~$p\ge 1$\/} if~$\psi$ is~$(p-1)$ times but not~$p$ times differentiable at~$\beta$.

Let us now assume that the measured space~$S$ is endowed with a topology turning it into a compact Hausdorff space, so that the Hamiltonian~$H_\Lambda\colon\Omega_\Lambda\to\R$ is continuous with respect to the corresponding product topology on~$\Omega_\Lambda=S^\Lambda$. For any~$u\in\R$, consider the subspace 
\begin{align*}
M_\Lambda(u) &\coloneqq \{\omega\in\Omega_\Lambda\colon H_\Lambda(\omega)\leq |\Lambda|u\}\,.
\end{align*}
Note that this space is closed in the compact space~$\Omega_\Lambda$, and therefore itself compact.

As mentioned in the introduction, the idea of the topological hypothesis is to relate a phase transition at inverse temperature~$\beta$ with a change in the topology of~$M_\Lambda(u)$ at the corresponding energy~$u=\psi'(\beta)$, for~$|\Lambda|\to\infty$. In its strongest form, it asserts that this change in topology is apparent in a very coarse topological invariant, namely the Euler characteristic.

Recall that if a topological space~$M$ is (of the homotopy type of) a finite CW-complex, then its {\em Euler characteristic} is defined as
\begin{align*}
\chi(M)\coloneqq\sum_{i\ge 0}(-1)^i \, |\{i-\text{dimensional cells of~$M$}\}|\,.
\end{align*}
It is a remarkable fact that this integer does not depend on the cellular structure on~$M$, but only on its homotopy type (see~\cite[Chapter~2]{Hat02}). Note that in the case of a finite discrete space, the Euler characteristic is nothing but the cardinality of the underlying set.

Let us now assume that for each~$u$, the compact space~$M_\Lambda(u)$ has the homotopy type of a finite CW-complex. This assumption is quite natural: for example, it is satisfied whenever~$S$ is a compact manifold and~$H_\Lambda$ a Morse function on the manifold~$\Omega_\Lambda$ (see~\cite{Mil63}). Then, one can define the logarithmic density of the Euler characteristic of~$M_\Lambda(u)$ as 
\begin{align*}
\sigma(u)\coloneqq\lim_{|\Lambda|\to\infty}\frac{1}{|\Lambda|}\log|\chi(M_\Lambda(u))|\,.
\end{align*}

We are finally ready to formulate precisely the topological hypothesis.

\begin{hypo}
\label{th}
There is a phase transition at inverse temperature~$\beta=s'(u)$ if and only if the function~$\sigma$ is not smooth at~$u=\psi'(\beta)$.
\end{hypo}

This statement corresponds to Conjectures~V.1 and~VII.2 of~\cite{Kas08a}, and can be thought of as the strongest among the many incarnations of the topological hypothesis. The aim of the present note is to study its validity for a wide class of {\em discrete\/} lattice spin models.

\section{Equivalence of ensembles for lattice models}
\label{sec:ens_equiv}

In this section, we focus our attention on lattice models with Hamiltonian of a specific type, namely translation invariant and absolutely summable (Section~\ref{sub:gas}). For these models, thermodynamic equivalence of ensembles has been established in full mathematical rigour. In Section~\ref{sub:LPS}, we briefly recall these classical results, which will play a crucial role in Section~\ref{sec:main}. 

\subsection{Lattice spin models}
\label{sub:gas}

Let us start by recalling the general setting of lattice spin models, referring to~\cite{Geo11} for a more complete and formal description.

Let the spin space~$S$ be a compact Hausdorff space endowed with its Borel~$\sigma$-algebra and a finite measure. For any finite subset~$\Lambda$ of~$\Z^d$, we shall write~$(\Omega_\Lambda,\mathcal{F}_\Lambda,\rho_\Lambda)$ for the corresponding product measured space, and denote by~$\omega=(\omega_x)_{x\in\Lambda}$ the elements of~$\Omega_\Lambda=S^\Lambda$.

Fix an {\em interaction potential\/}~$\Phi=\{\Phi_A\}$, i.e. an~$\mathcal{F}_A$-measurable function~$\Phi_A\colon\Omega_A\to\R$ for each non-empty finite subset~$A$ of~$\Z^d$. We will assume {\em translation invariance\/} of this (interaction) potential, a fact formalised by the equality
\[
\vartheta_x\Phi_A=\Phi_{\vartheta_x A}\colon\Omega_{\vartheta_x A}\to\R
\]
for all~$x\in\Z^d$ and finite~$A\subset\Z^d$, where~$\vartheta_x A=\{y+x\colon y\in A\}$ and~$\vartheta_x\Phi_A(\omega)=\Phi_A(\vartheta_x\omega)$ with~$(\vartheta_x\omega)_y=\omega_{y-x}$ for~$\omega\in\Omega_{\vartheta_x A}$ and~$y\in A$. We will also assume this potential to be {\em absolutely summable\/}, i.e. to satisfy
\begin{align*}
 \|\Phi\| &\coloneqq \sum_{A\ni 0}\|\Phi_A\|< \infty\,,
\end{align*}
where~$\|\Phi_A\|=\sup_{\omega\in\Omega_A}|\Phi_A(\omega)|$. The associated Hamiltonian~$H_\Lambda\colon\Omega_\Lambda\to\R$ is defined by
\[
H_\Lambda(\omega)=\sum_{A\subset\Lambda}\Phi_A(\omega_A)\,,
\]
where~$\omega_A$ denotes the restriction of~$\omega\in\Omega_\Lambda$ to~$\Omega_A$. In this setting, the quantity~$h_\Lambda(\omega)\coloneqq\frac{H_\Lambda(\omega)}{|\Lambda|}$ is called the {\em energy per site\/} of~$\omega\in\Omega_\Lambda$.

Let us illustrate these concepts with a classical example.

\begin{ex}
\label{ex:Ising}
Consider the spin set~$S=\{-1,1\}$ endowed with the discrete topology and the counting measure. Fix a family of real {\em coupling constants\/}~$(J_{x,y})$ indexed by~$\{x,y\}\subset\Z^d$ with~$x\neq y$ together with a real-valued {\em magnetic field\/}~$(h_x)_{x\in\Z^d}$. Define the potential~$\Phi=\{\Phi_A\}$ by
\[
\Phi_A(\omega)=\left\{\begin{array}{lll}
     -J_{x,y}\,\omega_x\omega_y & \text{for~$A=\{x,y\}\subset\Z^d$ with~$x\neq y$},\\
     -h_x\,\omega_x & \text{for~$A=\{x\}\subset\Z^d$,}\\
     0 & \text{else.}
        \end{array}\right.
\]
This potential is translation invariant if and only if~$J_{x,y}=J_{0,y-x}$ and~$h_x=h_0$ for all~$x,y\in\Z^d$, and absolutely summable exactly when~$\sum_{x\in\Z^d}|J_{0,x}|$ is finite. The associated Hamiltonian is
\[
H_\Lambda(\omega)=-\sum_{\{x,y\}\subset\Lambda,\,x\neq y}J_{x,y}\,\omega_x\omega_y\,-\,\sum_{x\in\Lambda}h_x\,\omega_x\,.
\]
The resulting model is the celebrated {\em Ising model\/} on~$\Z^d$. It is called {\em ferromagnetic\/} if~$J_{x,y}\ge 0$ for all~$x,y\in\Z^d$, {\em finite-range\/} if there exists~$R>0$ such that~$J_{x,y}=0$ for all~$x,y\in\Z^d$ with~$|x-y|>R$, and {\em nearest neighbour\/} if it is finite-range with~$R=1$.
\end{ex}

\subsection{Equivalence of ensembles}
\label{sub:LPS}
We now state in a precise way the equivalence of ensembles in the general setting of Section~\ref{sub:gas}. These type of results are classical, going back to the early days of rigorous statistical mechanics~\cite{Rue68,Lan73}.

For definiteness, let~$\Lambda_n$ denote the hypercube~$[-n,n]^d\cap\mathbb{Z}^d$. We use the shorthand notation~$(\Omega_n,\mathcal{F}_n,\rho_n)$ for the corresponding sequence of measured spaces and~$H_n(\omega),h_n(\omega)$ for the energy and energy per site of~$\omega\in\Omega_n$.
Since the potential is translation invariant and absolutely summable, all the maps~$h_n$ take values in the compact interval~$I:=[-\|\Phi\|,\|\Phi\|]$.

Let~$\mathbb{M}_n$ denote the finite measure on the Borel sets of~$I$ given by~$\mathbb{M}_n\coloneqq \rho_n\circ h_n^{-1}$. In other words, we set
\[
\mathbb{M}_n(B) = \rho_n\left\{\omega\in\Omega_n\colon \frac{H_n(\omega)}{|\Lambda_n|}\in B\right\}
\]
for any Borel subset~$B$ of~$I$.

\begin{prop}
\label{prop:equiv}
\begin{enumerate}[(i)]
\item For any interval~$B\subset I$, the limit
\[
m(B)\coloneqq \lim_{n\to\infty}\frac{1}{|\Lambda_n|}\log(\mathbb{M}_n(B))
\]
exists in~$\overline{\mathbb{R}}\coloneqq\R\cup\{-\infty,\infty\}$.
\item For all~$u\in I$, the limit
\[
s(u)\coloneqq\lim_{r\to 0}m\left(\left(u-r,u+r\right)\right)
\]
exists, defining a concave function~$s\colon I\to\overline{\mathbb{R}}$.
\item For any interval~$B\subset I$, we have~$m(B) = \sup_{x\in B}s(x)$.
\item The pressure~$\psi$ and the entropy~$s$ are Legendre-Fenchel duals.
\end{enumerate}
\end{prop}

As mentioned above, these results have their origins in the pioneering work of Ruelle~\cite{Rue68} and Lanford~\cite{Lan73}. In the case of discrete spin models (which is the only case we will use), these statements can be found in Sections~2 and~3 of~\cite{ML79}. In the (perhaps too) general setting of Section~\ref{sub:gas}, they follow from Corollary~3.1, Corollary~5.1 and Lemma~5.2 of~\cite{Lew94} (see also~\cite{Lew95}).

\section{The topological hypothesis for discrete spin models}\label{sec:main}

In this section, we state and prove our main results which deal with {\em discrete\/} spin models. We begin in Section~\ref{sub:sigma} by showing that the function~$\sigma$ coincides with the entropy~$s$ for positive temperatures. In Section~\ref{sub:conv}, we prove that under some hypothesis on the potential, the negative\footnote{Recall the unconventional sign in our definition of the pressure.} of the pressure is strongly convex. Section~\ref{sub:main} contains our main result, which can be considered as some modified version of the topological hypothesis valid for a wide class of discrete models. Finally, in Section~\ref{sub:Ising}, we show that the original topological hypothesis (Hypothesis~\ref{th}) holds for the ferromagnetic nearest neighbour Ising model on~$\Z^d$, with the possible exception of dimensions~$3$ and~$4$.

Throughout this section, we assume that the spin space~$S$ is finite, endowed with the discrete topology and the counting measure.

\subsection{The function~$\sigma$ and the entropy}
\label{sub:sigma}

We begin this section with an easy but fundamental result: while~$s$ is defined as Boltzmann's surface entropy, the function~$\sigma$ can be understood as Gibbs' volume entropy in the case of discrete models.

To make this statement precise, let us assume that the effective domain of~$s$, defined as~$\{u\in I\colon s(u)>-\infty\}$, consists of a non-empty open interval~$(a,c)\subset I$. (In degenerate cases, it could be reduced to a point.) Let us also denote by~$b\in(a,c)$ a real number where the concave function~$s$ reaches its maximum.

\begin{lem}
\label{lemma:sigma}
The logarithmic density of the Euler characteristic~$\sigma(u)$ exists for all~$u>a$; it coincides with~$s(u)$ for~$u\in(a,b]$ and is equal to~$s(b)=\log(|S|)$ for~$u\ge b$ (see Figure~\ref{fig:3}).
\end{lem}

\begin{figure}
\centering
\includegraphics[width=.6\textwidth]{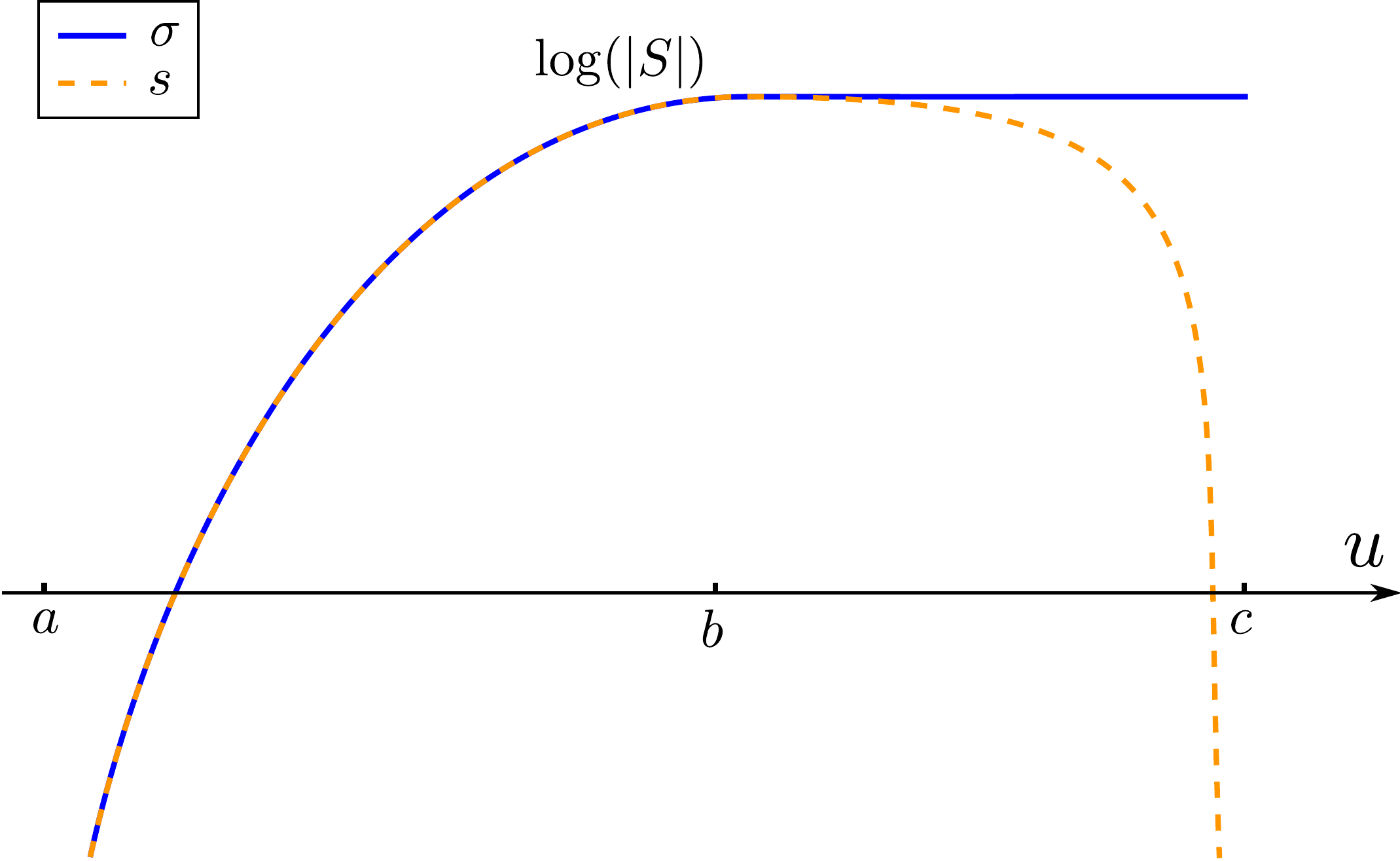}
\caption{Plot of the logarithmic density of the Euler characteristic~$\sigma$ and of the entropy~$s$. The intervals~$(a,b)$ and~$(b,c)$ correspond to positive and negative temperatures, respectively.}
\label{fig:3}
\end{figure}
 
\begin{proof}
Recall that the function~$\sigma$ is defined by~$\sigma(u)=\lim_{n\to\infty}\frac{1}{|\Lambda_n|}\log|\chi(M_n(u))|$, where
\[
M_n(u)=\{\omega\in\Omega_n\colon H_n(\omega)\le |\Lambda_n|\,u\}=h_n^{-1}((-\infty,u])=h_n^{-1}([-\|\Phi\|,u])\,.
\]
Since~$M_n(u)$ is finite and discrete, its Euler characteristic is simply its cardinality. The measure~$\rho_n$ being the counting measure, we obtain
\[
\chi(M_n(u))=|h_n^{-1}([-\|\Phi\|,u])|=\rho_n\circ h_n^{-1}([-\|\Phi\|,u])=\mathbb{M}_n([-\|\Phi\|,u])\,.
\]
By Proposition~\ref{prop:equiv}, we have
\[
\sigma(u)=\lim_{n\to\infty}\frac{1}{|\Lambda_n|}\log\mathbb{M}_n([-\|\Phi\|,u])=m([-\|\Phi\|,u])=\sup_{x\in[-\|\Phi\|,u]}s(x)\in\overline{\R}
\]
for all~$u\in I=[-\|\Phi\|,\|\Phi\|]$. Since~$s$ is concave and finite on~$(a,c)$ and reaches its maximum at~$b\in(a,c)$, it follows that~$\sigma(u)=s(u)$ for all~$u\in(a,b]$ and~$\sigma(u)=s(b)$ for~$u\ge b$. The definition of~$\sigma$ implies that~$\sigma(u)=\log(|S|)$ for~$u\ge\|\Phi\|$, concluding the proof.
\end{proof}

Although very elementary, this observation is already a significant step towards the topological hypothesis for discrete models. Indeed, as the function~$\sigma$ coincides with the entropy on the interval~$(a,b)$, it is the Legendre-Fenchel dual of the pressure~$\psi$ restricted to positive temperatures. Hence, a phase transition at some inverse temperature~$\beta_c>0$ is likely to correspond to a non-smooth point~$u_c\in(a,b)$ of~$\sigma=s$, and vice-versa.

However, the situation is not as simple in general. As an easy counterexample, consider the pressure given by~$\psi(\beta)=u_c\beta-\frac{3}{4}\left|\beta-\beta_c\right|^{4/3}$. This function is not twice differentiable at~$\beta_c$, while its Legendre-Fenchel dual~$\sigma$ satisfies~$\sigma'(u)=\beta_c-(u-u_c)^3$, and is therefore smooth (with~$\sigma''(u_c)=0$).
The reverse phenomenon could a priori also happen, namely the existence of a non-smooth point~$u_c$ of~$\sigma$ that is not reflected by any phase transition, but only by the second derivative of~$\psi$ vanishing at the corresponding~$\beta_c$.

Therefore, more work is required to prove the topological hypothesis for discrete models. This is the aim of the next section.

\subsection{Strong convexity of the pressure}
\label{sub:conv}

The negative of the pressure as defined in Section~\ref{sub:equiv} is the limit of convex functions, so it is always convex. For some general class of models, it can be shown to be strictly convex (see~\cite{GR71} and~\cite[Corollary~16.15]{Geo11}). Unfortunately, this does not imply that~$\psi''(\beta)$ never vanishes when defined, a condition needed for our main result to hold. For this, we need the notion of ``strong convexity''.

Recall that a map~$f\colon(a,b)\to\R$ is {\em strongly convex\/} with parameter~$c>0$ if
\begin{equation*}
f(tx+(1-t)y)\le tf(x)+(1-t)f(y)-\frac{c}{2}t(1-t)|x-y|^2
\end{equation*}
for all~$x,y\in(a,b)$ and~$t\in[0,1]$. Note that this condition is equivalent to the function~$g\colon(a,b)\to\R$ defined by~$g(x)=f(x)-\frac{c}{2}x^2$ being convex. In particular, this implies the inequality~$f''(x)\ge c>0$ for all~$x\in(a,b)$ such that~$f''(x)$ exists.

To ensure that~$-\psi$ is strongly convex, we will require the potential~$\Phi=\{\Phi_A\}$ to be {\em non-constant\/}, meaning that there exists~$A\subset\Z^d$ with~$\Phi_A$ non-constant. This condition is clearly necessary: if all~$\Phi_A$ are constant, then the pressure is an affine function (given by~$-\psi(\beta)=\log(|S|)-\beta\sum_{A\ni 0}\Phi_A$) and therefore not strongly convex.

We will also require the potential to be {\em positively correlated\/}, in the sense that
\[
\cov_{\Lambda,\beta}(\Phi_A,\Phi_B)\coloneqq\left<\Phi_A\Phi_B\right>_{\Lambda,\beta}-\left<\Phi_A\right>_{\Lambda,\beta}\left<\Phi_B\right>_{\Lambda,\beta}\ge 0
\]
for all~$A,B\subset\Lambda$ and all~$\beta>0$. Here, we use the customary notation~$\left<f\right>_{\Lambda,\beta}$ for the expected value of the function~$f\colon\Omega_\Lambda\to\R$ with respect to the Gibbs distribution~$\mu_{\Lambda,\beta}$ on~$\Omega_\Lambda$, that is
\begin{align*}
\left<f\right>_{\Lambda,\beta}&\coloneqq\sum_{\omega\in\Omega_\Lambda}f(\omega)\frac{e^{-\beta H_\Lambda(\omega)}}{Z_{\Lambda,\beta}}\, , &Z_{\Lambda,\beta}\coloneqq\sum_{\omega\in\Omega_\Lambda}e^{-\beta H_\Lambda(\omega)}\,.
\end{align*}
Also, we make a slight abuse of notation and use the symbol~$\Phi_A$ both for the map~$\Phi_A\colon\Omega_A\to\R$ and for its extension~$\Omega_\Lambda\to\R$ given by~$\omega\mapsto\Phi_A(\omega_A)$, where~$\omega_A$ denotes the restriction of~$\omega\in\Omega_\Lambda$ to~$\Omega_A$.

Let us illustrate this condition with an example.

\begin{ex}
\label{ex:Grifiths}
Fix a positive integer~$k$ and set~$S=\{-k,-k+2,\dots,k-2,k\}$. Let~$\Phi=\{\Phi_A\}$ be the potential given by~$\Phi_A=-J_A\sigma_A$, where~$\sigma_A(\omega)=\prod_{x\in A}\omega_x$ and~$J_A$ is a non-negative real number. Then, for any~$A,B\subset\Lambda$ and~$\beta>0$, we have
\[
\cov_{\Lambda,\beta}(\Phi_A,\Phi_B)=\left<\Phi_A\Phi_B\right>_{\Lambda,\beta}-\left<\Phi_A\right>_{\Lambda,\beta}\left<\Phi_B\right>_{\Lambda,\beta}=J_AJ_B\left(\left<\sigma_A\sigma_B\right>_{\Lambda,\beta}-\left<\sigma_A\right>_{\Lambda,\beta}\left<\sigma_B\right>_{\Lambda,\beta}\right)\ge 0
\] 
by Griffiths' second inequality~\cite{Gri69}. Therefore, this potential is positively correlated. This holds in particular for the ferromagnetic Ising model (with~$h\ge 0$), which corresponds to the case~$k=1$ and~$J_A=0$ for~$|A|>2$.
\end{ex}

Let us quickly mention other natural classes of examples. If~$S$ is a finite abelian group and~$-\Phi_A\colon\Omega_A\to\R$ is a positive definite function for all~$A\subset\Z^d$, then the potential~$\Phi=\{\Phi_A\}$ is positively correlated by Ginibre's inequality, see~\cite[Example~4]{Gin70}. (Note that the case~$S=\Z_2$ and~$-\Phi_A=J_A\sigma_A$ with~$J_A\ge 0$ once again corresponds to the Ising model.) Also, if~$S$ is a finite distributive lattice and all the maps~$\Phi_A$ are ``submodular'' and monotone increasing (or all monotone decreasing), then~$\Phi=\{\Phi_A\}$ is positively correlated by the FKG inequality~\cite{FKG71}.

\medskip

We are ready to state the main result of this section.

\begin{prop}
\label{prop:conv}
Consider a lattice spin model with finite spin space endowed with the counting measure. Assume that the potential is translation invariant, absolutely summable, non-constant and positively correlated. Then, for any bounded interval~$(a,b)\subset(0,\infty)$, there exists~$c>0$ such that~$-\psi\colon(a,b)\to\R$ is strongly convex with parameter~$c$. In particular, the second derivative of~$\psi$ is strictly negative whenever defined.
\end{prop}

We will need one preliminary result.

\begin{lem}
\label{lem:var}
Let~$S$ be finite and endowed with the counting measure and let~$\Phi$ be translation invariant and absolutely summable. If~$A\subset\Z^d$ is such that~$\Phi_A$ is non-constant, then there exists a continuous map~$c\colon[0,\infty)\to(0,\infty)$ such that~$\var_{\Lambda,\beta}(\Phi_A)\ge c(\beta)$ for all~$\beta\ge 0$ and all~$\Lambda$ containing~$A$.
\end{lem}

\begin{proof}[Proof of Lemma~\ref{lem:var}]
As a first step, let us show that for any~$\lambda\in\Phi_A(\Omega_A)$ and any~$\Lambda$ containing~$A$, there exists a continuous map~$c_\lambda\colon[0,\infty)\to(0,\infty)$, independent of~$\Lambda$, such that
\[
\mu_{\Lambda,\beta}(\Phi_A=\lambda)\ge c_\lambda(\beta)\tag{$\star$}
\]
for all~$\beta\ge 0$. To check this claim, let us fix~$\overline{\omega}\in\Omega_{\Lambda\setminus A}$ and decompose the Hamiltonian as
\[
H_\Lambda(\omega)=\Phi_A(\omega_A)+\sum_{\substack{B\subset\Lambda\,,B\neq A\\B\cap A\neq \emptyset}}\Phi_B(\omega_B)+\sum_{\substack{C\subset\Lambda\\C\cap A=\emptyset}}\Phi_C(\omega_C)\,.
\]
Since the potential is translation invariant, the second term is bounded by
\[
\Big|\sum_{\substack{B\subset\Lambda\,,B\neq A\\B\cap A\neq \emptyset}}\Phi_B(\omega_B)\Big|\le\sum_{\substack{B\subset\Lambda\\B\cap A\neq\emptyset}}\|\Phi_B\|\le |A|\sum_{B\ni 0}\|\Phi_B\|=|A|\,\|\Phi\|\,,
\]
which is finite since~$\Phi$ is absolutely summable. Therefore, writing~$K(\beta)$ for~$e^{\beta|A|\|\Phi\|}$ and~$c(\overline{\omega})$ for~$e^{-\beta\sum_{C\subset\Lambda\,,C\cap A=\emptyset}\Phi_C(\omega_C)}$, we have the inequalities
\[
e^{-\beta\Phi_A(\omega_A)}\,K(\beta)^{-1}\,c(\overline{\omega})\le e^{-\beta H_\Lambda(\omega)}\le e^{-\beta\Phi_A(\omega_A)}\,K(\beta)\, c(\overline{\omega})\,.
\]
Using the notation~$\Omega_\Lambda^{\overline{\omega}}=\{\omega\in\Omega_\Lambda\,|\,\omega_{\Lambda\setminus A}=\overline{\omega}\}$, it follows that
\begin{align*}
\mu_{\Lambda,\beta}(\Phi_A=\lambda\,|\,\omega_{\Lambda\setminus A}=\overline{\omega})&=\frac{\sum_{\omega\in\Omega_\Lambda^{\overline{\omega}}\,,\Phi_A(\omega_A)=\lambda}e^{-\beta H_\Lambda(\omega)}}{\sum_{\omega\in\Omega_\Lambda^{\overline{\omega}}}e^{-\beta H_\Lambda(\omega)}}\\
&\ge\frac{\sum_{\omega\in\Omega_\Lambda^{\overline{\omega}}\,,\Phi_A(\omega_A)=\lambda}e^{-\beta\lambda}\,K(\beta)^{-1}\,c(\overline{\omega})}{\sum_{\omega\in\Omega_\Lambda^{\overline{\omega}}}e^{-\beta\Phi_A(\omega_A)}\,K(\beta)\,c(\overline{\omega})}\\
&=\frac{|\{\omega\in\Omega_A\,|\,\Phi_A(\omega)=\lambda\}|}{K(\beta)^{2}\sum_{\omega\in\Omega_A}e^{-\beta(\Phi_A(\omega)-\lambda)}}=:c_\lambda(\beta)\,.
\end{align*}
Since the map~$c_\lambda$ defined by the last equality is continuous, positive, and depends neither on~$\Lambda$ nor on~$\overline{\omega}$, the inequality~$(\star)$ follows and the claim is proved.

Since~$\Phi_A\colon\Omega_A\to\R$ is not constant, there exists~$\lambda_1\neq\lambda_2$ in~$\Phi_A(\Omega_A)$. Hence, we have the inequalities
\begin{align*}
\var_{\Lambda,\beta}(\Phi_A)&=\sum_{\omega\in\Omega_\Lambda}\left(\Phi_A(\omega_A)-\left<\Phi_A\right>_{\Lambda,\beta}\right)^2\mu_{\Lambda,\beta}(\omega)\\
&=\sum_{\lambda\in\Phi_A(\Omega_A)}\left(\lambda-\left<\Phi_A\right>_{\Lambda,\beta}\right)^2\mu_{\Lambda,\beta}(\Phi_A(\omega_A)=\lambda)\\
&\stackrel{(\star)}{\ge}\left(\lambda_1-\left<\Phi_A\right>_{\Lambda,\beta}\right)^2 c_{\lambda_1}(\beta)+\left(\lambda_2-\left<\Phi_A\right>_{\Lambda,\beta}\right)^2 c_{\lambda_2}(\beta)\\
&\ge\frac{1}{2}(\lambda_1-\lambda_2)^2\min\{c_{\lambda_1}(\beta),c_{\lambda_2}(\beta)\}=:c(\beta)\,,
\end{align*}
and the lemma is proved.
\end{proof}

\begin{proof}[Proof of Proposition~\ref{prop:conv}] By definition, the pressure is equal to~$\psi(\beta)=\lim_{|\Lambda|\to\infty}\psi_\Lambda(\beta)$, with~$\psi_\Lambda(\beta)=-\frac{1}{|\Lambda|}\log Z_{\Lambda,\beta}$ and~$Z_{\Lambda,\beta}=\sum_{\omega\in\Omega_\Lambda}e^{-\beta H_\Omega(\omega)}$. Direct computations give
\[
-\psi_\Lambda''(\beta)=\frac{1}{|\Lambda|}\left(\left<H_\Lambda^2\right>_{\Lambda,\beta}-\left<H_\Lambda\right>^2_{\Lambda,\beta}\right)=\frac{1}{|\Lambda|}\var_{\Lambda,\beta}(H_\Lambda)=\frac{1}{|\Lambda|}\var_{\Lambda,\beta}\big(\sum_{A\subset\Lambda}\Phi_A\big)\,.
\]
Since the potential is assumed to be positively correlated, we have
\[
\var_{\Lambda,\beta}\big(\sum_{A\subset\Lambda}\Phi_A\big)=\sum_{A\subset\Lambda}\var_{\Lambda,\beta}(\Phi_A)+\sum_{\substack{A,B\subset\Lambda\\A\neq B}}\cov_{\Lambda,\beta}(\Phi_A,\Phi_B)\ge\sum_{A\subset\Lambda}\var_{\Lambda,\beta}(\Phi_A)\,.
\]
By assumption, there exists~$A_0\subset\Z^d$ such that~$\Phi_{A_0}$ is not constant. Assuming that~$\Lambda$ is the union of translated copies of~$A_0$ (i.e. of subsets of~$\Z^d$ of the form~$A_0+x$ with~$x\in\Z^d$), translation invariance of the potential now implies
\[
-\psi''_\Lambda(\beta)\ge\frac{1}{|\Lambda|}\sum_{\substack{A\subset\Lambda\\A=A_0+x}}\var_{\Lambda,\beta}(\Phi_A)=\var_{\Lambda,\beta}(\Phi_{A_0})\frac{|\{A\subset\Lambda\,|\,A=A_0+x\}|}{|\Lambda|}\ge\frac{\var_{\Lambda,\beta}(\Phi_{A_0})}{|A_0|}\,.
\]
By Lemma~\ref{lem:var}, we conclude that there exists a continuous map~$c\colon[0,\infty)\to(0,\infty)$, independent of~$\Lambda$, such that~$-\psi''_\Lambda(\beta)\ge c(\beta)$ for all~$\beta\ge 0$. This implies that the function~$-\psi_\Lambda$ is strongly convex on any bounded interval~$(a,b)\subset[0,\infty)$, with parameter~$\min_{\beta\in[a,b]}c(\beta)>0$. Since this parameter is independent of~$\Lambda$, the same holds true for the limit~$-\psi(\beta)=\lim_{|\Lambda|\to\infty}-\psi_\Lambda(\beta)$. This concludes the proof.
\end{proof}

\subsection{The topological hypothesis for discrete spin models}
\label{sub:main}

We are finally ready to state our main result, whose proof is now straightforward.

\begin{thm}
\label{thm:main}
Consider a lattice spin model with finite spin space endowed with the discrete topology and the counting measure. Assume that the potential is translation invariant, absolutely summable, non-constant and positively correlated, and that the system does not exhibit any first-order phase transition.

Then, there exists~$a<b\in\R$ such that the following statements hold:
\begin{enumerate}[(i)]
\item{The pressure~$\psi\colon(0,\infty)\to\R$ and entropy~$s\colon(a,b)\to\R$ are differentiable and Legendre-Fenchel duals, so~$\psi'$ and~$s'$ are mutually inverse continuous maps.}
\item{The function~$\sigma$ coincides with~$s$ on~$(a,b)$.}
\end{enumerate}
Furthermore, for any~$\beta=s'(u)>0$:
\begin{enumerate}[(i)]
\setcounter{enumi}{2}
\item{The system undergoes a second order phase transition at~$\beta$ if and only if~$\sigma$ is not twice differentiable at~$u$ or~$\sigma''(u)=0$.}
\item{The system undergoes a phase transition of order~$p>2$ at~$\beta$ if and only if~$\sigma$ is~$(p-1)$ but not~$p$ times differentiable at~$u$ and~$\sigma''(u)\neq 0$.}
\end{enumerate}
\end{thm}

\begin{proof}
Since the potential is translation invariant and absolutely summable, Proposition~\ref{prop:equiv} states that the entropy~$s\colon I\to\overline{\R}$ and the pressure~$\psi\colon\R\to\R$ are Legendre-Fenchel duals. By hypothesis,~$\psi$ is differentiable, hence continuously differentiable since it is concave. By Proposition~\ref{prop:conv}, it is also strictly concave on~$\R$. This implies that its dual~$s=\psi^*$ is (continuously) differentiable on its effective domain~$\{u\in I\colon s(u)>-\infty\}$ which consists of a non-empty open interval~$(a,c)\subset I$ (see e.g.~\cite{WV73}). Therefore, the maps~$\psi'\colon\R\to(a,c)$ and~$s'\colon(a,c)\to\R$ are strictly decreasing continuous functions which are mutual inverses. In particular, the real number~$a$ (resp.~$c$) is nothing but the limit of~$\psi'(\beta)$ as~$\beta$ tends to~$\infty$ (resp. to~$-\infty$). Writing~$b$ for~$\psi'(0)\in(a,c)$, the first point is proved. Note that~$s'(b)=(\psi')^{-1}(b)=0$, so~$s$ has a unique maximum at~$u=b$. We are therefore in the setting of Lemma~\ref{lemma:sigma}, which implies the second point.

As a consequence of points~$(i)$ and~$(ii)$, the continuous maps~$\sigma'=s'\colon(a,b)\to(0,\infty)$ and~$\psi'\colon(0,\infty)\to(a,b)$ are mutual inverses, with~$\psi''(\beta)$ nowhere zero by Proposition~\ref{prop:conv}. This easily implies points~$(iii)$ and~$(iv)$, as we now demonstrate. Fix~$\beta>0$ and set~$u:=s'(\beta)=\sigma'(\beta)\in(a,b)$. If the system undergoes a second order phase transition at~$\beta$, then~$\psi'$ is not differentiable at~$\beta$; since~$\psi'$ and~$\sigma'$ are inverses, either~$\sigma'$ is not differentiable at~$u$ or~$\sigma''(u)$ vanishes. Conversally, if~$\psi'$ is differentiable at~$\beta$, then~$\sigma'$ is differentiable at~$u$ since~$\psi''(\beta)$ does not vanish. Furthermore, the chain rule applied to~$\psi'\circ\sigma'=\mathit{id}$ leads to the equality~$\psi''(\beta)\sigma''(u)=1$, so~$\sigma''(u)$ does not vanish either. This shows point~$(iii)$. Finally, if the system undergoes a phase transition of order~$p>2$ at~$\beta$, then~$\psi'$ is~$(p-2)$ but not~$(p-1)$ times differentiable at~$\beta$ and its derivative does not vanish at~$\beta$; by the inverse function theorem, the function~$\sigma'$ has the same properties. Exchanging the roles of~$\psi'$ and~$\sigma'$ concludes the proof.
\end{proof}

\subsection{The Ising model}
\label{sub:Ising}

As a motivating example, we now apply Theorem~\ref{thm:main} to the ferromagnetic Ising model on~$\mathbb{Z}^d$. Our understanding of this model depends greatly on the dimension, so we shall present the results in the form of a discussion culminating in the main statement: the validity of the original topological hypothesis for the nearest neighbour ferromagnetic Ising model on~$\mathbb{Z}^d$, with the possible exceptions of dimensions~$d=3,4$ (Theorem~\ref{thm:Ising}).

As usual, we shall assume throughout this section that the coupling constants~$(J_{x,y})_{x,y\in\Z^d}$ are translation invariant, absolutely summable and ferromagnetic (recall Example~\ref{ex:Ising}), but also satisfy the following property: for all~$x\in\Z^d$, there exists~$0=x_0,\dots,x_m=x$ such that~$J_{x_0,x_1}\cdots J_{x_{m-1},x_m}>0$. We also fix a magnetic field~$h\in\R$. Note that the pressure~$\psi$ is unchanged when replacing~$h$ with~$-h$, so we can assume~$h\ge 0$ without loss of generality.

The potential corresponding to these coupling constants and magnetic field is translation invariant, absolutely summable and non-constant by assumption, and positively correlated by Example~\ref{ex:Grifiths} (recall that~$J_{x,y}\ge 0$ and~$h\ge 0$). Furthermore, by~\cite[Corollary~2]{Rao17} (see also~\cite{Aiz15}), the four conditions on the coupling constants stated above imply that the model does not undergo a first-order phase transition. Therefore, the hypothesis of Theorem~\ref{thm:main} (and Proposition~\ref{prop:conv}) are satisfied, so~$\psi'\colon(0,\infty)\to(a,b)$ and~$\sigma'\colon(a,b)\to(0,\infty)$ are mutual inverses with~$\psi''$ never vanishing. (For the Ising model, one easily checks that~$b=\psi'(0)=0$.)

We now start the aforementioned case by case discussion.

\subsubsection*{Non-vanishing magnetic field}
Let us first assume that the magnetic field~$h\in\R$ is non-zero. Then, by~\cite[p.~109]{LP68}, the pressure~$\psi$ is analytic on~$(0,\infty)$. Since~$\sigma'$ and~$\psi'$ are inverse with~$\psi''\neq 0$, it follows that~$\sigma$ is analytic on~$(a,0)$. Therefore, Hypothesis~\ref{th} holds (trivially) in this case.

From now on, we assume that the magnetic field is equal to zero.

\subsubsection*{The critical inverse temperature}
For a wide class of Ising models, including the ones under study in this section, there exists a {\em critical inverse temperature\/}~$\beta_c\in[0,\infty]$ so that the {\em spontaneous magnetization\/}~$\left<\sigma_0\right>^+_\beta$ vanishes for~$\beta<\beta_c$ while~$\left<\sigma_0\right>^+_\beta>0$ for~$\beta>\beta_c$. Here,~$\left<\sigma_0\right>^+_\beta$ denotes the expected value of~$\sigma_0$ with respect to the infinite volume Gibbs measure with plus boundary condition (see~\cite{FV17}). 

The pressure~$\psi$ is expected to be analytic on~$(0,\infty)\setminus\{\beta_c\}$, with the {\em specific heat\/}~$-\psi''(\beta)$ exhibiting a special type of singularity at~$\beta_c$ (see e.g.~\cite{FFS92}, and details below). As we shall see, this would imply the validity of the topological hypothesis. More precisely, we could conclude that~$\sigma$ is analytic on~$(a,0)\setminus\{u_c\}$ and not smooth at~$u_c\coloneqq\psi'(\beta_c)$. However, these facts are proven only in some cases, as we now explain.

\subsubsection*{Dimension one}
Let us consider the Ising model on~$\Z$. In the finite-range case, we have~$\beta_c=\infty$ and the pressure is known to be analytic on~$(0,\infty)$ (see e.g.~\cite{Rue68}). It follows that~$\sigma$ is analytic on~$(a,0)$ and Hypothesis~\ref{th} is valid. The same result is expected to hold for coupling constants satisfying~$\sum_{x\in\Z}xJ_{0,x}^2<\infty$ (see~\cite{Rue68}).

In the remaining cases, i.e. when coupling constants decay as~$J_{x,y}\sim |x-y|^{-\alpha}$ with~$1<\alpha\le 2$, the critical inverse temperature is known to be finite and strictly positive~\cite{Dys69,ACCN88}. However, the behavior of~$\psi''$ at this critical point seems unknown, and we cannot conclude that Hypothesis~\ref{th} holds.

\subsubsection*{Dimension two} 
Consider the two-dimensional nearest neighbour Ising model, with coupling constants~$J_1$ and~$J_2$. In a classical work, Onsager~\cite{Ons44} was able to compute the pressure as
\[
-\psi(\beta)=\log 2+\frac{1}{2\pi^2}\int_0^{\pi}\!\!\int_0^{\pi}\log P(\theta_1,\theta_2)d\theta_1d\theta_2\,,
\]
where
\[
P(\theta_1,\theta_2)=
\cosh(2\beta J_1)\cosh(2\beta J_2)-\sinh(2\beta J_1)\cos\theta_1-\sinh(2\beta J_2)\cos\theta_2\,.
\]
This leads to the identification of the critical inverse temperature~$\beta_c$ as the unique positive solution to the equation~$\sinh(2\beta J_1)\sinh(2\beta J_2)=1$. In the case~$J_1=J_2=1$, the solution is given by~$\beta_c=\frac{1}{2}\log(\sqrt{2}+1)$, a value first predicted in~\cite{KW41}.

With the explicit expression above,~$\psi$ is easily shown to be smooth at~$\beta\neq\beta_c$, with the specific heat having a logarithmic singularity at~$\beta_c$. Therefore, the map~$\psi'$ has a singularity of the form~$\psi'(\beta)\sim(\beta-\beta_c)\log|\beta-\beta_c|$ at~$\beta_c$. Since the derivative of~$\psi'$ never vanishes, the inverse map~$\sigma'=(\psi')^{-1}\colon(a,0)=(-J_1-J_2,0)\to(0,\infty)$ is smooth at all~$u\neq u_c=\psi'(\beta_c)=J_1\cosh(2\beta_cJ_2)$, twice differentiable at~$u_c$ (with~$\sigma''(u_c)=0$), but not three times differentiable at~$u_c$. In particular, Hypothesis~\ref{th} holds.

Note that the same analysis can be performed for any biperiodic planar graph (see~\cite{CD13}), and the same conclusion holds. However, the analyticity of~$\psi$ seems unknown in the general (i.e. not nearest neighbour) case.

\subsubsection*{Smoothness of the pressure outside the critical point} 
In the subcritical regime~$\beta<\beta_c$, exponential decay of the two-point correlation functions~$\left<\sigma_0\sigma_x\right>_\beta$ has been established in~\cite{ABF87} for finite-range models (see also~\cite{DT16} for an alternative proof). By~\cite[p.~318]{Leb72}, this implies that the pressure is smooth for all~$\beta<\beta_c$. (Note however that this is not sufficient to conclude that the pressure is analytic.)

In the supercritical regime~$\beta>\beta_c$, exponential decay of the truncated two-point correlation functions~$\left<\sigma_0\sigma_x\right>_\beta-\left<\sigma_0\right>_\beta\left<\sigma_x\right>_\beta$ has been recently proved for finite-range models of dimension~$d\ge 3$, see~\cite{DGR18}. Again, by the argument of~\cite{Leb72}, this shows that the pressure is smooth for all~$\beta>\beta_c$.

In conclusion, the pressure is smooth at~$\beta\neq\beta_c$ for finite-range models in dimension~$d\ge 3$. By Theorem~\ref{thm:main}, this implies that~$\sigma$ is smooth at all~$u\neq u_c=\psi'(\beta_c)$. To show that Hypothesis~\ref{th} holds, it remains to understand the specific heat near the critical point.

\subsubsection*{Specific heat in dimension~$d\ge 3$} 
For nearest neighbour models in dimension~$d>4$, the specific heat~$-\psi''(\beta)$ is known to be uniformly bounded~\cite{Sok79}. As a consequence, since~$\psi'$ and~$\sigma'$ are mutual inverses, the second derivative~$\sigma''(u)$ never vanishes. By Theorem~\ref{thm:main}, it follows that Hypothesis~\ref{th} holds in this case:~$\sigma$ is smooth at all~$u\neq u_c$ and smooth at~$u_c$ if and only if~$\psi$ is smooth at~$\beta_c$.

Note that the specific heat is expected to exhibit a jump discontinuity at~$\beta_c$ (see~\cite[p.~281]{FFS92}). This would imply that~$\sigma$ also has a jump discontinuity at~$u_c$, but no proof of this statement is currently available.

In dimension~$d=4$, the critical exponent~$
\alpha\coloneqq\lim_{\beta\to\beta_c}-\frac{\log|\psi''(\beta)|}{\log|\beta-\beta_c|}$ is known to vanish~\cite{Sok79}. Furthermore, the specific heat is conjectured to exhibit a logarithmic singularity at~$\beta_c$ (see~\cite{FFS92,LM09}). Hypothesis~\ref{th} would then hold, but this has not yet been formally established.

Finally, very little is known in dimension~$d=3$. Numerical experiments~\cite{BCGS94} give the approximative value~$\alpha\approx 0.104$. Having~$-\psi''(\beta)\sim|\beta-\beta_c|^{-\alpha}$ with~$\alpha\approx 0.104$ suggests~$\sigma'(u)\sim-|u-u_c|^{\frac{1}{1-\alpha}}$ with~$\frac{1}{1-\alpha}\approx 1.116$. If rigorously established, this would imply that~$\sigma$ is twice but not three times differentiable at~$u_c$ and would confirm the validity of Hypothesis~\ref{th} is this dimension as well.

\bigskip

As a consequence of the above discussion, we have proved Hypothesis~\ref{th} for the nearest neighbour ferromagnetic Ising model on~$\Z^d$ in all dimensions except~$d=3,4$. More precisely:

\begin{thm}
\label{thm:Ising}
Consider the translation invariant nearest neighbour ferromagnetic Ising model on~$\Z^d$ with non-identically zero coupling constants and arbitrary magnetic field. Then, the pressure~$\psi\colon(0,\infty)\to\R$ is smooth at all~$\beta\neq\beta_c$ and the function~$\sigma\colon(a,0)\to\R$ is smooth at all~$u\neq u_c=\psi'(\beta_c)$. Furthermore,~$\sigma$ is not smooth at~$u_c$ if and only if~$\psi$ is not smooth at~$\beta_c$, with the possible exception of dimensions~$3$ and~$4$.\qed
\end{thm}

\bigskip

We conclude this note with one last comment. For some discrete spin models, the topological hypothesis does not hold in any possible sense. As an easy example of this fact, consider the {\em Curie-Weiss model\/} defined by the spin space~$S=\{-1,1\}$ and the Hamiltonian~$H_\Lambda(\omega)=-\frac{1}{|\Lambda|}\sum_{x,y\in\Lambda}\omega_x\omega_y$. This model is well-known to undergo a phase-transition at~$\beta_c=\frac{1}{2}$ (see e.g~\cite[Chapter~2]{FV17}). However, a direct computation shows that the function~$\sigma$ is constant (equal to~$\log(2)$). Therefore, the non-analytic behavior of the pressure is not reflected in any way in~$\sigma$.

\bibliographystyle{plain}

\end{document}